\documentclass[11pt]{article}
\usepackage{amsmath,amsthm,amssymb}
\usepackage{fullpage}
\usepackage{lmodern}
\usepackage{amsmath,amsfonts}
\usepackage{graphicx}
\usepackage{dsfont}
\usepackage[colorlinks=true,linkcolor=blue]{hyperref}
\usepackage{doi}
\usepackage{microtype}

\usepackage[numbers]{natbib}
\newcommand{\normal}{\lhd}
\newcommand{\eqdef}{\triangleq}

\newcommand{\CC}{\mathbb{C}}
\newcommand{\N}{\mathbb{N}}

\DeclareMathOperator{\positive}{\textbf{P}}

\DeclareMathAlphabet{\varmathbb}{U}{bbold}{m}{n}
\newcommand{\one}{\mathds{1}}
\DeclareMathOperator*{\Exp}{\mathbb{E}}

\DeclareMathOperator{\Res}{Res}
\begin{document}

\newtheorem*{lemma-restatement}{Lemma}
\newtheorem*{theorem-restatement}{Theorem}
\newtheorem{theorem}{Theorem}
\newtheorem{lemma}[theorem]{Lemma}
\newtheorem{claim}[theorem]{Claim}
\newtheorem{proposition}[theorem]{Proposition}
\newtheorem{corollary}[theorem]{Corollary}
\theoremstyle{definition}
\newtheorem{definition}{Definition}
\newtheorem{construction}{Construction}
\newtheorem{fact}[theorem]{Fact}

\newcommand{\eps}{\varepsilon}
\newcommand{\Z}{{\mathbb{Z}}}
\newcommand{\C}{{\mathbb{C}}}
\newcommand{\R}{{\mathbb{R}}}
\newcommand{\F}{{\mathbb{F}}}
\newcommand{\U}{{\textsf{U}}}
\newcommand{\Var}{\mathop\mathrm{Var}}
\renewcommand{\Re}{\mathop\mathrm{Re}}
\newcommand{\frob}[1]{\left\| #1 \right\|_{\rm F}}
\newcommand{\opnorm}[1]{\left\| #1 \right\|_{\rm op}}
\newcommand{\bopnorm}[1]{\bigl\| #1 \bigr\|_{\rm op}}
\newcommand{\norm}[1]{\| #1 \|}
\newcommand{\bnorm}[1]{\bigl\| #1 \bigr\|}
\newcommand{\Bnorm}[1]{\Bigl\| #1 \Bigr\|}
\newcommand{\abs}[1]{\left| #1 \right|}
\newcommand{\inner}[2]{\left\langle #1, #2 \right\rangle}
\newcommand{\poly}{\mathrm{poly}}

\newcommand{\wf}{\widehat{f}}
\newcommand{\wg}{\widehat{g}}
\newcommand{\wG}{\widehat{G}}
\newcommand{\wH}{\widehat{H}}
\newcommand{\wpsi}{\widehat{\psi}}
\newcommand{\wpsidag}{\widehat{\psi^\dagger}}
\newcommand{\GL}{\textsf{GL}}
\newcommand{\fsub}{\operatorname{F}}

\renewcommand{\vec}[1]{\mathbf{#1}}

\title{Small-Bias Sets for Nonabelian Groups:\\ Derandomizing the Alon-Roichman Theorem}

\author{Sixia Chen \and Cristopher Moore \and Alexander Russell}

\maketitle

\abstract{In analogy with $\eps$-biased sets over $\Z_2^n$, we
  construct explicit $\eps$-biased sets over nonabelian finite groups
  $G$.  That is, we find sets $S \subset G$ such that $\norm{\Exp_{x
      \in S} \rho(x)} \le \eps$ for any nontrivial irreducible
  representation $\rho$.  Equivalently, such sets make $G$'s Cayley
  graph an expander with eigenvalue $|\lambda| \le \eps$.  
  The Alon-Roichman theorem shows that random sets of size $O(\log |G| / \eps^2)$ suffice.  
  For groups of the form $G = G_1 \times \cdots \times G_n$, our construction has
  size $\poly(\max_i |G_i|, n, \eps^{-1})$, and we show that
  a set $S \subset G^n$ considered by Meka and Zuckerman that fools
  read-once branching programs over $G$ is also $\eps$-biased in this sense.  
  For solvable groups  whose abelian quotients have constant exponent, we obtain $\eps$-biased sets of size 
  $(\log |G|)^{1+o(1)} \,\poly(\eps^{-1})$.  
%  nearly-polynomial size, $\poly((\log |G|)^{\log \log \log |G|} , \eps^{-1})$.  
  Our techniques include derandomized squaring (in both
  the matrix product and tensor product senses) and a Chernoff-like
  bound on the expected norm of the product of independently random
  operators that may be of independent interest.}

\section{Introduction}

Small-bias sets are useful combinatorial objects for derandomization,
and are particularly well-studied over the Boolean hypercube $\{0,1\}^n$.  Specifically, if we identify the hypercube with the group $\Z_2^n$, then a \emph{character} $\chi$ is a homomorphism from $\Z_2^n$ to $\C$.  We say that a set $S \subseteq \F_2^n$ is \emph{$\eps$-biased} if, for all characters $\chi$,
\[
\abs{ \Exp_{x \in S} \chi(x) } \le \eps \, ,
\]
except for the trivial character $\one$, which is identically equal to $1$.  Since any character of $\F_2^n$ can be written $\chi(x) = (-1)^{k \cdot x}$ where $k \in \Z_2^n$ is the ``frequency vector,''  
%(in which case $\chi(x)=1$ or $-1$ depending on $x$'s parity on the support of $k$), 
this is equivalent to the familiar definition which demands that on any nonzero set of bits, $x$'s parity should be odd or even with roughly equal probability, $(1 \pm \eps)/2$.

It is easy to see that $\eps$-biased sets of size $O(n / \eps^2)$ exist: random sets suffice.  Moreover, several efficient deterministic constructions are
known~\cite{NaorN:Small,AlonBNNR:Construction,DBLP:journals/rsa/AlonGHP92,BenAroyaTS:Constructing} of size polynomial in $n$ and $1/\eps$.  These constructions have been used to derandomize a wide variety of randomized algorithms, replacing random sampling over all of $\{0,1\}^n$ with deterministic sampling on $S$ (see e.g.~\cite{BogdanovV:Pseudorandom}).  In particular, sampling a function on an $\eps$-biased set yields a good estimate of its expectation if its Fourier spectrum has bounded $\ell_1$ norm.

The question of whether similar constructions exist for
nonabelian groups has been a topic of intense interest.  Given a group $G$, a \emph{representation} is a homomorphism $\rho$ from $G$ into the group $\U(d)$ of $d \times d$ unitary matrices for some $d=d_\rho$.  If $G$ is finite, then up to isomorphism there is a finite set $\wG$ of \emph{irreducible} representations, or \emph{irreps} for short, such that any representation $\sigma$ can be written as a direct sum of irreps.  These irreps form the basis for harmonic analysis over $G$, analogous to classic discrete Fourier analysis on abelian groups such as $\Z_p$ or $\Z_2^n$.

Generalizing the standard notion from characters to matrix-valued representations, we say that a set $S \subseteq G$ is \emph{$\eps$-biased} if, for all nontrivial irreps $\rho \in \wG$, 
\[
\Bnorm{\Exp_{x \in S} \rho(x) } \le \eps \, ,
\]
where $\norm{\cdot}$ denotes the operator norm.  There is a natural connection with expander graphs.  If we define a Cayley graph on $G$ using $S$ as a set of generators, then $G$ becomes an expander if and only if $S$ is $\eps$-biased.  Specifically, if $M$ is the stochastic matrix equal to $1/|S|$ times the adjacency matrix, corresponding to the random walk where we multiply by a random element of $S$ at each step, then $M$'s second eigenvalue has absolute value $\eps$.  Thus $\eps$-biased sets $S$ are precisely sets of generators that turn $G$ into an expander of degree $|S|$.
%Classic results then show, for instance, that the mixing time of this random walk is $O(\log |G|)$.

The Alon-Roichman theorem~\cite{AlonR:Random} asserts that a uniformly random set of $O((\log |G|) / \eps^2)$ group elements is $\eps$-biased with high probability.  Thus, our goal is to
derandomize the Alon-Roichman theorem---finding explicit constructions 
%, for as many families of nonabelian groups as possible
of $\eps$-biased sets of size polynomial in $\log |G|$ and $1/\eps$.  (For another notion of derandomizing the Alon-Roichman theorem, in time $\poly(|G|)$, see Wigderson and Xiao~\cite{WigdersonX:Derandomizing}.)

Throughout, we apply the technique of ``derandomized
squaring''---analogous to the principal construction in Rozenman and
Vadhan's alternate proof of Reingold's theorem~\cite{Rozenman:2005fk}
that Undirected Reachability is in \textsf{LOGSPACE}. In particular,
we observe that derandomized squaring provides a generic amplification tool in our setting;  specifically, given a constant-bias set $S$, we can obtain an $\eps$-biased set of size $O(|S| \eps^{-11})$.  We also use a tensor product version of derandomized squaring to build $\eps$-biased sets from $G$ recursively, from $\eps$-biased sets for its subgroups or quotients.
%Thus we focus on achieving constant bias.

\paragraph{Homogeneous direct products and branching programs} 
Groups of the form $G^n$ where $G$ is fixed have been actively studied by the
  pseudorandomness community as a specialization of the class of
  constant-width branching programs. The problem of fooling  ``read-once'' group programs induces an alternate notion of
  $\eps$-biased sets over groups of the form $G^n$ defined by Meka and
  Zuckerman~\cite{Meka:2009fk}.  Specifically, a read-once branching
  program on $G$ consists of a tuple $\vec{g} = (g_1,\ldots,g_n) \in
  G^n$ and takes a vector of $n$ Boolean variables $\vec{b} =
  (b_1,\ldots,b_n)$ as input.  At each step, it applies $g_i^{b_i}$, i.e., $g_i$ if $b_i=1$ and $1$ if $b_i=0$.  
 They say a set $S \subset G^n$ is $\eps$-biased if, for all
  $\vec{b} \ne 0$, the distribution of $\vec{g}^\vec{b}$ is close to
  uniform, i.e.,
  \begin{equation}\label{prop:MZ}
  \forall h \in G: 
  \abs{
    \Pr_{\vec{g} \in S} \left[ \vec{g}^\vec{b} = h \right] - \frac{1}{|G|} 
  } \le \eps 
  \quad \text{where} \quad 
\vec{g}^\vec{b} = \prod_{i=1}^n g_i^{b_i} \, .
  \end{equation}
  As they comment, there is no obvious relationship
%, in either direction,
  between this definition and the one we consider.\footnote{In particular,
  there is no obvious way to amplify in their setting: for instance,
  squaring a set $S$ by multiplication in $G^n$ squares the operator
  norm of any representation, but it has a very complicated effect on
  the distribution of $\vec{g}^\vec{b}$.}  We are unable to establish
  such a connection in general.  However, we show in
  Section~\ref{sec:homogeneous} that a particular set shown to have property~\eqref{prop:MZ}
in~\cite{Meka:2009fk} is also $\eps$-biased in our sense; the proof
  is completely different. This yields $\eps$-biased sets of size
  $O(n \cdot \poly(\eps^{-1}))$.
%  ; see Section~\ref{sec:homogeneous}
%  and Theorem~\ref{thm:homogeneous}.

  \paragraph{Inhomogeneous direct products} For the more general case of groups of the form $G = G_1
  \times \cdots \times G_n$, we show that a tensor product adaptation of derandomized squaring yields a recursive
  construction of $\eps$-biased sets of size $\poly(\max_i |G_i|, n,
  1/\eps)$. 
  %See Section~\ref{sec:direct}, Theorem~\ref{thm:direct}.

\paragraph{Normal extensions and ``smoothly solvable'' groups}  
Finally, we show that if $G$ is solvable and has abelian quotients of bounded exponent, we can construct $\eps$-biased sets of size $(\log |G|)^{1+o(1)} \,\poly(\eps^{-1})$.  
%Specifically, if the length of $G$'s derived series is $\ell$, our sets have size $O((\log |G|)^{\log \log \ell}) = O((\log |G|)^{\log \log \log |G|})$.  
Here we use the representation theory of solvable groups to build an $\eps$-biased set for $G$ recursively from those for a normal subgroup $H$ and the quotient $G/H$.

\section{An explicit set for $G^n$ with constant $\eps$}
\label{sec:homogeneous}

Meka and Zuckerman~\cite{Meka:2009fk} considered the following construction for fooling
read-once group branching programs:
\begin{definition}\label{def:MZ}
  Let $G$ be a group and $n \in \N$. Then, given an 
  $\eps$-biased set $S$ over $\Z_{|G|}^n$, define
\[
T_S \eqdef\{(g^{s_1},\dots,g^{s_n} )\mid  g\in G, (s_1,\dots,s_n)\in
S\} \, .
\]
\end{definition}

We prove the following theorem, showing that this construction
yields sets of small bias in our sense (and, hence, expander Cayley
graphs over $G^n$).
\begin{theorem}\label{thm:constant}
If $S$ is $\eps$-biased over $\mathbb{Z}_{|G|}^n$ then $T_S$ is $(1
- \Omega(1/\log\log |G|)^2 + \eps)$-biased over $G^n$.
\end{theorem}
\noindent
Anticipating the proof, we set down the following definition.
\begin{definition}
  Let $G$ be a finite group. For a representation $\rho \in
  \widehat{G}$ and a subgroup $H$, define
  \[
  \Pi_H^\rho \triangleq 
  %\frac{1}{|H|} \sum_{h \in H} \rho(h)
  \Exp_{h \in H} \rho(h)
  \]
  to be the projection operator induced by the subgroup $H$ in $\rho$.
  In the case where $H=\langle g \rangle$ is the cyclic 
  group generated by $g$, we use the following shorthand:
  \[
  \Pi^\rho_g = \Pi^\rho_{\langle g \rangle} \, . 
  \]
Finally, for groups of the form $G^n$ we use the following convention.  Recall that any irreducible representation $\bar \rho \in \widehat{G^n}$ is a tensor product, $\bar \rho = \bigotimes_{i=1}^n \rho_i$ where $\rho_i \in \wG$ for each $i$.  That is, if $\bar{g} = (g_1,\ldots,g_n)$, then $\bar \rho(\bar{g}) = \bigotimes_{i=1}^n \rho_i(g_i)$.  Then for an element $g \in G$, we write
  \begin{equation}\label{def:proj}
  \Pi_g^{\bar \rho} \triangleq \Pi_{\langle g\rangle^n}^{\bar \rho} =
  \bigotimes_{i=1}^n \Pi^{\rho_i}_g 
  %{\langle g\rangle}
  \end{equation}
  for the projection operator determined by the abelian subgroup
  $\langle g \rangle^n$.
\end{definition}

\begin{lemma}\label{lem:power-avg}
  Let $G$ be a finite group and $\rho$ a nontrivial irreducible
  representation of $G$. Then
  \[
  \Bnorm{ \Exp_{g \in G}   \Pi_g^\rho } 
  \leq 1 - \frac{\phi(|G|)}{|G|} \leq 1 - \Omega\left(\frac{1}{\log \log |G|} \right) \, ,
  \]
  where $\phi(\cdot)$ denotes the Euler totient function.
\end{lemma}

\begin{proof}
  Expanding the definition of $\Pi_{\langle g\rangle}^\rho$, we have
  \[
  \Bnorm{ \Exp_{g \in G}   \Pi_g^\rho } 
  = \Bnorm{ \Exp_{g \in G} \Exp_{t \in \Z_{|G|}} \rho(g^t) } 
  \leq \Exp_{t \in \Z_{|G|}} \Bnorm{ \Exp_g \rho(g^t) } \, .
  \]
  Recall that the function $x \mapsto x^k$ is a bijection in any
  group $G$ for which $\gcd(|G|, k) = 1$.  Moreover, for such $k$, $\Exp_g
  \rho(g^k) = \Exp_g \rho(g) = 0$ as $\rho \neq 1$. Assuming
  pessimistically that $\norm{\Exp_g \rho(g^k)} = 1$ for all other
  $k$ yields the bound $\norm{ \Exp_{g \in G} \Pi_{\langle g\rangle}^\rho } \leq
  1 - \phi(|G|)/|G|$ promised in the statement of the lemma. The
  function $\phi(n)$ has the property that
  \[
  {\phi(n)} > \frac{n}{e^\gamma \log\log n + \frac{3}{\log \log n}}
  \]
  for $n > 3$, where $\gamma \approx .5772\ldots$ is the Euler
  constant~\cite{RosserS1962}; this yields the second estimate in the
  statement of the lemma.
\end{proof}

Our proof will rely on the following tail bound for products of
operator-valued random variables, proved in Appendix~\ref{app:tail}.

\begin{theorem}\label{thm:operator-avg}
  Let $\positive(H)$ denote the cone of positive operators on the
  Hilbert space $H$.  Let $P_1,\ldots, P_k$ be independent random
  variables taking values in $\positive(H)$ for which $\norm{ P_i } \le 1$ 
  and $\bnorm{ \Exp[P_i] } \leq 1 - \delta$.
  Then 
  \[
  \Pr\left[ \bnorm{ P_k \cdots P_1} \geq
    \sqrt{\dim H} \exp\left(-\frac{k\delta}{6}\right)\right] \leq
  \dim H \cdot \exp\left(-\frac{k\delta^2}{13} \right) \, . \qedhere
%  {18 \ln 2}
  \]
\end{theorem}

We return to the proof of Theorem~\ref{thm:constant}.

\begin{proof}[Proof of Theorem~\ref{thm:constant}]
For a non-trivial irrep $\bar \rho=\rho_1 \otimes \cdots \otimes \rho_n \in
\wG^n$, we write
\[
\Exp_{\bar t\in T_S}\;\bar \rho(\bar t) 
= \Exp_{g \in G} \Exp_{\bar s \in S}\;\bar\rho(g^{\bar s}) 
= \Exp_{g \in G} \Exp_{\bar s \in S}\;\left( \Res_{\langle g \rangle^n} \bar\rho \right)(g^{\bar s}) \, ,
\]
where $\bar s=(s_1, \dots, s_n)$, $g^{\bar s}=(g^{s_1},\dots,
g^{s_n})$, and $\Res_H \bar{\rho}$ denotes the restriction of $\bar{\rho}$ to the subgroup $H \subseteq G^n$. 
For a particular $g \in G$, we decompose the restricted representation $\Res_{\langle g \rangle^n} \bar\rho$ 
into a direct sum of irreps of the abelian group $\langle g \rangle^n \cong \Z_{|\langle g \rangle|}^n$.  This yields
%; this yields an expansion
\[
\Res_{\langle g\rangle^n} \bar\rho 
= \bigoplus_{\chi \in \widehat{\langle g \rangle^n}} \chi^{\oplus a_\chi} \, , 
\]
where each $\chi$ is a one-dimensional representation of the cyclic 
group $\langle g\rangle^n$ and $a_\chi$ denotes the multiplicity with which $\chi$ appears in the
decomposition. 

Now, as $S$ is an $\eps$-biased set over $\Z_{|G|}^n$, its
quotient modulo any divisor $d$ of $|G|$ is
$\eps$-biased over $\Z_d^n$. It follows that 
\[
\left| \Exp_{\bar s \in S} \chi(\bar s) \right| \leq \eps
\]
for any nontrivial $\chi$; when $\chi$ is trivial, the expectation is 1.  
Thus for any fixed $g \in G$ we may write
\[
\Exp_{\bar s\in S}\; \left( \Res_{\langle g \rangle^n} \bar\rho\right)(g^{\bar s}) 
= \Pi_{g}^{\bar \rho} + E_{g}^{\bar \rho} \, .
\]
Recall that $\Pi_g^{\bar \rho}$ is the projection operator onto the space associated with the copies of the trivial representation of 
$\langle g \rangle^n$ in $\Res_{\langle g\rangle^n} \bar\rho$, i.e., the expectation we would obtain if $\bar{s}$ ranged over all of 
$\langle g \rangle^n$ instead over just $S$.  The ``error operator'' $E_{g}^{\bar\rho}$ arises from the nontrivial representations 
of $\langle g \rangle^n$ appearing in $\Res_{\langle g\rangle^n} \bar\rho$, and has operator norm bounded by $\eps$.  
%: $\| E_g^{\bar\rho}\| \leq \eps$.
It follows that
\begin{align*}
\Bnorm{ \Exp_{\bar t \in T} \bar\rho(\bar t) } 
&= \Bnorm{ \Exp_{g \in G} \left(\Exp_{\bar s \in S}  \bar\rho(g^{\bar s})\right) }
= \Bnorm{ \Exp_{g \in G} \left( \Pi_g^{\bar\rho} + E_g^{\bar\rho}\right) } \\
& \leq \Bnorm{ \Exp_{g \in G} \Pi_g^{\bar\rho} } 
+ \Bnorm{ \Exp_{g \in G} E_g^{\bar\rho} } 
\leq \Bnorm{ \Exp_{g \in G} \Pi_g^{\bar\rho} } + \eps \, ,
\end{align*}
and it remains to bound $\norm{ \Exp_{g\in G} \Pi^{\bar \rho}_g }$.

As $\Exp_{g\in G} \Pi^{\bar \rho}_g$ is Hermitian, for any positive $k$ we have
\begin{equation}
\label{eq:power}
\Bnorm{ \Exp_{g\in G} \Pi^{\bar \rho}_g } 
= \sqrt[k]{\left\| \left(\Exp_{g\in G} \Pi^{\bar \rho}_g\right)^k \right\| } \, ,
\end{equation}
so we focus on the operator $\left(\Exp_{g\in G} \Pi^{\bar \rho}_g\right)^k$.
Expanding $\Pi^{\bar \rho}_g = \bigotimes_i \Pi^{\rho_i}_g$, we may
write
\begin{equation}\label{eq:expansion}
\left(\Exp_{g\in G} \Pi^{\bar \rho}_g\right)^k = \Exp_{g_1, \ldots,
  g_k} \left[ \Pi_{g_1}^{\bar \rho} \cdots \Pi_{g_k}^{\bar \rho}\right] = \Exp_{g_1, \ldots,
  g_k} \left[\bigotimes_{i=1}^n \Pi_{g_1}^{\rho_i} \cdots \Pi_{g_k}^{\rho_i}\right]  \, .
\end{equation}
As $\bar\rho$ is nontrivial, there is some coordinate $j$ for which
$\rho_j$ is nontrivial.  Combining~\eqref{eq:expansion} with the fact
that $\norm{A \otimes B} = \norm{A} \norm{B}$, we conclude that
\begin{equation}\label{eq:triangle}
\left\| \left(\Exp_{g\in G} \Pi^{\bar \rho}_g \right)^k \right\| 
\leq \Exp_{g_1, \ldots, g_k} \left\| \bigotimes_{i=1}^n \Pi_{g_1}^{\rho_i} \cdots \Pi_{g_k}^{\rho_i} \right\|
\leq \Exp_{g_1,\ldots, g_k} \left\| \Pi^{\rho_j}_{g_1} \cdots \Pi^{\rho_j}_{g_k} \right\| \, .
\end{equation}
Lemma~\ref{lem:power-avg} asserts that $\norm{\Exp_g \Pi_g^{\rho_j} } \leq 1 - \delta_G$, 
where $\delta_G = \Omega(1 / \log \log |G|)$. It follows then from
Theorem~\ref{thm:operator-avg} that
\begin{equation}\label{eq:tail-bound}
\Pr_{g_1, \ldots, g_k} \Biggl[\underbrace{\Bnorm{ \Pi^{\rho_j}_{g_1} \cdots \Pi^{\rho_j}_{g_k} } 
\geq \sqrt{d_j} \exp(-k \delta_G/6)}_{(\ddag)} \Biggr] \leq d_j \cdot \exp\left(-{k \delta_G^2}/{13}\right) \, ,
\end{equation}
where $d_j = \dim \rho_j$. This immediately provides a bound on $\norm{ (\Exp_g \Pi_g^{\bar \rho})^k }$.  Specifically,
combining~\eqref{eq:triangle} with~\eqref{eq:tail-bound}, let us pessimistically assume that 
$\norm{ \Pi^{\rho_j}_{g_1}\cdots \Pi^{\rho_j}_{g_k} } = d_j \exp(-k\delta_G/6)$
for tuples $(g_1, \ldots, g_k)$ that do not enjoy property $(\ddag)$, and $1$ for tuples that do. Then
 \begin{align*}
 \left\| \left(\Exp_{g\in G} \Pi^{\bar \rho}_g\right)^k \right\|
 & \leq \Exp_{g_1,\ldots, g_k} \bnorm{ \Pi^{\rho_j}_{g_1} \ldots \Pi^{\rho_j}_{g_k} } \\
 &\leq d_j \exp\left(-{k
 \delta_G^2}/{13}\right) + \left(1 - d_j \exp\left(-{k
 \delta_G^2}/{13}\right)\right) \sqrt{d_j} \exp\left(-{k
   \delta_G}/{6}\right)\\
 &\leq 2d_j \exp\left(-{k \delta_G^2}/{13}\right) \, , 
 \end{align*}
and hence
\[
\left\| \Exp_{g \in G} \Pi_g^{\bar \rho}\right\|
 \leq \inf_k \left( \sqrt[k]{2d_j}\right) \cdot \exp(-{\delta_G^2}/{13}) 
 = 1 - \Omega\left({1}/{\log\log |G|}\right)^2 \, , %\qedhere
\]
where we take the limit of large $k$.
\end{proof}

\section{Derandomized squaring and amplification} 
% of $\eps$-biased sets over groups}
\label{sec:amplification}

In this section we discuss how to amplify $\eps$-biased sets in a generic way.  Specifically, we use derandomized squaring to prove the following.  
\begin{theorem}
\label{thm:amplify}
Let $G$ be a group and $S$ an $1/10$-biased set on $G$.  Then for any $\eps > 0$, there is an $\eps$-biased set $S_\eps$ on $G$ of size $O(|S| \eps^{-11})$.  Moreover, assuming that multiplication can be efficiently implemented in $G$, the set $S_\eps$ can be constructed from $S$ in time polynomial in $|S_\eps|$.
\end{theorem}
\noindent
We have made no attempt to improve the exponent of $\eps$ in $|S_\eps|$.

Our approach is similar to~\cite{Rozenman:2005fk}.  Roughly, if $S$ is an $\eps$-biased set on $G$ we can place a degree-$d$ expander graph $\Gamma$ on the elements of $S$ to induce a new set
\[
S \times_\Gamma S \eqdef \{ st \mid \text{$(s, t)$ an edge of $\Gamma$}\} \, .
\]
If $\rho: G \rightarrow \textsf{U}(V)$ is a nontrivial representation of $G$, by
assumption $\|\Exp_{s \in S} \rho(g) \| \leq \eps$. Applying
a natural operator-valued Rayleigh quotient for expander graphs (see
Lemma~\ref{lem:expander-operators} below), we conclude that
\[
\left\|\Exp_{(s,t) \in \Gamma} \rho(s)\rho(t)\right\| = \left\|\Exp_{(s,t) \in \Gamma}
\rho(st)\right\| \leq \lambda(\Gamma) + \eps^2 \, .
\]
If $\Gamma$ comes from a family of Ramanujan-like expanders, then
$\lambda(\Gamma) = \Theta(1/\sqrt{d})$, and we can guarantee that
$\lambda(\Gamma) = O(\eps^2)$ by selecting $d = \Theta(\eps^{-4})$.
The size of the set then grows by a factor of $|S \times_\Gamma S| /
|S| = d = \Theta(\eps^{-4})$.  We make this precise in
Lemma~\ref{lem:amplify} below, which regrettably loses an additional
factor of $\eps^{-1}$.

Preparing for the proof of Theorem~\ref{thm:amplify}, we record some
related material on expander graphs.

\paragraph{Expanders and derandomized products}

For a $d$-regular graph $G = (V, E)$, let $A$
denote its normalized adjacency matrix: 
$A_{uv} = 1/d$ if $(u,v) \in E$ and $0$ otherwise.  
%\[
%[A_G]_{st} = \begin{cases} 1/d &\text{if $(s,t) \in E$,}\\
%0 & \text{otherwise.} \end{cases}
%\]
Then $A$ is stochastic, normal, and has operator norm $\norm{A} =
1$; the uniform eigenvector $\vec{y^+}$ given by $y^+_s = 1$ for all $s \in V$ 
has eigenvalue $1$.  When $G$ is connected, the
eigenspace associated with $1$ is spanned by this eigenvector, and
all other eigenvalues lie in $[-1, 1)$.

Bipartite graphs will play a special role in our analysis. We write a
bipartite graph $G$ on the bipartition $U, V$ as the tuple $G = (U, V; E)$. In
a regular bipartite graph, we have $|U| = |V|$ and $-1$ is an eigenvalue of
$A$ associated with the eigenvector $\vec{y^-}$ which is $+1$ for $s \in U$ and $-1$ for $s \in V$. 
%given by 
%\[
%\quad y_s^- = \begin{cases} +1 & \text{if $s \in
%    U$,}\\
%-1 & \text{if $s \in V$.}
%\end{cases}
%\]
When $G$ is connected, the eigenspace associated with $-1$
is one-dimensional, and all other eigenvalues lie in $(-1, 1)$: we let 
$\lambda(G) < 1$ be the leading nontrivial eigenvalue:
\[
\lambda(G) = \sup_{\vec{y} \; \perp\; {\vec{y^{\pm}}}} {\| M\vec{y} \|}/{\| \vec{y} \|} \, .
\]
When $\vec{y} \perp \vec{y^{\pm}}$, observe that $|{\langle \vec{y}, M\vec{y} \rangle}| \leq \| \vec{y} \| \cdot \| M\vec{y}
\| \leq \lambda \| \vec{y}\|^2$ by Cauchy-Schwarz.

We say that a $d$-regular, connected, bipartite graph $G = (U, V; E)$
for which $|U| = |V| = n$ and $\lambda(G) \leq \Lambda)$ is a \emph{bipartite $(n, d, \Lambda)$-expander}.  A well-known consequence of expansion is that the ``Rayleigh quotient'' determined by the expander is bounded: for
any function $f: U \cup V \rightarrow \R$ defined on the vertices of a
$(n, d, \lambda)$ expander for which $\sum_{u \in U} f(u) = \sum_{v \in V} f(v) = 0$,
\[
\Exp_{(u,v) \in E} f(u)f(v) \leq \lambda \|f\|_2^2 \, .
\]
We will apply a version of this property pertaining to operator-valued functions.

\begin{lemma}
\label{lem:expander-operators} 
Let $G = (U , V; E)$ be a bipartite $(n, d, \lambda)$-expander.  Associate with each vertex $s \in U \cup V$ a linear operator $X_s$ on the vector space $\CC^d$ such that $\norm{X_s} \leq 1$, $\bnorm{ \Exp_{u \in U} X_u } \leq \eps_U$, and $\bnorm{ \Exp_{v \in V} X_v } \leq \eps_V$.  Then
  \[
  \Bnorm{ \Exp_{(u,v) \in E} X_u X_v } \leq \lambda + (1 - \lambda) \eps_U \eps_V \, .
  \]
\end{lemma}
\noindent
We will sometimes apply Lemma~\ref{lem:expander-operators} to the tensor product of operators.  That is, given the same assumptions, we have
  \[
  \Bnorm{ \Exp_{(u,v) \in E} X_u \otimes X_v } \leq \lambda + (1 - \lambda) \eps_U \eps_V \, .
  \]
To see this, simply apply the lemma to the operators $X_u \otimes \one$ and $\one \otimes X_v$.
  
Critical in our setting is the fact that this conclusion is
independent of the dimension $d$. A proof of this folklore lemma
appears in Appendix~\ref{appendix:expanders}; see
also~\cite{De:Pseudorandomness} for a related application to branching
programs over groups.

\paragraph{Amplification} We return now to the problem of amplifying
$\eps$-biased sets over general groups.

\begin{lemma}
\label{lem:amplify}
Let $S$ be an $\eps$-biased set on the group $G$.  Then there is an $\eps'$-biased set $S'$ on $G$ for which $\eps' \le 5 \eps^2$ and $|S'| \le  C |S| \eps^{-5}$, where $C$ is a universal constant. Moreover, assuming that multiplication can be efficiently implemented in $G$, the set $S'$ can be constructed from $S$ in time polynomial in $|S'|$.
\end{lemma}

\begin{proof}
  We proceed as suggested above. The only wrinkle is that we need to
  introduce an expander graph on the elements of $S$ that achieves
  second eigenvalue $\Theta(\eps^2)$.

  We apply the explicit family of Ramanujan graphs due to Lubotzky,
  Phillips, and Sarnak~\cite{Lubotzky:Ramanujan}.  For each pair of
  primes $p$ and $q$ congruent to $1$ modulo $4$, they obtain a graph $\Gamma_{p,q}$ with $p(p^2-1)$ vertices, degree $q+1$, and
  $\lambda(\Gamma_{p,q}) = 2\sqrt{q}/(q+1) < 2/\sqrt{q}$. We treat
  $\Gamma_{p,q}$ as a bipartite graph by taking the double cover: this
  introduces a pair of vertices, $v_\text{A}$ and $v_\text{B}$, for
  each vertex $v$ of $\Gamma_{p,q}$ and introduces an edge $(u_A,
  v_B)$ for each edge $(u,v)$. This graph has eigenvalues $\pm \lambda$ for each eigenvalue $\lambda$ of $\Gamma_{p,q}$, so except for the $\pm 1$ eigenspace the spectral radius is unchanged.

  As we do not have precise control over the number of vertices in
  this expander family, we will use a larger graph and approximately
  tile each side with copies of $S$. Specifically, we select the
  smallest primes $p,q \equiv 1 \pmod{4}$ for which
  \begin{equation}\label{eq:size-degree}
  p(p^2 - 1) > |S|
  \cdot \lceil \eps^{-1}\rceil\quad\text{and}\quad 2/\sqrt{q} \leq
  \eps^2 \, .
  \end{equation}
  We now associate elements of $S$ with the vertices (of each side) of
  $\Gamma = \Gamma_{p,q} = (U, V; E)$ as uniformly as possible;
  specifically, we partition the vertices of $U$ and $V$ into a family
  of blocks, each of size $|S|$; this leaves a set of less than $|S|$
  elements uncovered on each side. Then elements in the blocks are
  directly associated with elements of $S$; the ``uncovered'' elements
  may in fact be assigned arbitrarily. As $|U| = |V| \geq |S|
  \lceil\eps^{-1}\rceil$, the uncovered elements above comprise
  less than an $\eps$-fraction of the vertices. As above, we
  define the set $S \times_{\Gamma} S \eqdef \{ uv \mid (u,v) \in E\}$ (where we
  blur the distinction between a vertex and the element of $S$ to
  which it has been associated).

  Consider, finally, a nontrivial representation $\rho$ of $G$.  As the
  average over any block of $U$ or $V$ has operator norm no more
  than $\eps$, and we have an $\eps$-fraction of uncovered
  elements, the average of $\rho$ over each of $U$
  and $V$ is no more than $(1 - \eps)\eps + \eps \leq 2 \eps$. Applying Lemma~\ref{lem:expander-operators}, we conclude
  that $\| \Exp_{s \in S \times_\Gamma S} \rho(s) \| \leq
  (2\eps)^2 + \lambda(\Gamma) \leq 5 \eps^2$ by our choice of
  $q$ (the degree less one).

By Dirichlet's theorem on the density of primes in arithmetic progressions, $p$ and $q$ need be no more than (say) a constant factor larger than the lower bounds $p(p^2-1) > |S| \eps^{-1}$ and $q \ge 4 \eps^{-4}$ implied by~\eqref{eq:size-degree}.  Thus there is a constant $C$ such that 
$|S'| = p(p^2-1)(q+1) \le C |S| \cdot \eps^{-5}$.  
\end{proof}

\paragraph{Remarks} The construction above is saddled with the tasks of identifying appropriate primes $p$ and $q$, and constructing the generators for the associated expander of~\cite{Lubotzky:Ramanujan}.  While these can clearly be carried out in time polynomial in $|S'|$, alternate explicit constructions of expander graphs~\cite{Morgenstern:Existence} can significantly reduce this overhead. However, no known explicit family of Ramanujan graphs appears to provide enough density to avoid the tiling construction above. On the other hand, expander graphs with significantly weaker properties would suffice for the construction: any uniform bound of the form $\lambda \leq c\sqrt{\text{degree}}$ would be enough.

\begin{proof}[Proof of Theorem~\ref{thm:amplify}]
We apply Lemma~\ref{lem:amplify} iteratively.  Set $\eps_0 = 1/10$.  After $t$ applications, we have an $\eps_t$-biased set where $\eps_t = 2^{-2^t} / 5$.  After $t = \lceil \log_2 \log_2 (1/5 \eps) \rceil$ steps, we have $5 \eps^2 \le \eps_t \le \eps$.  The total increase in size is 
\[
\frac{|S_\eps|}{|S|} 
%= C^t (\eps_0 \eps_1 \ldots \eps_{t-1})^{-5} 
= C^t \left( \prod_{i=0}^{t-1} \eps_i \right)^{\!-5} 
= C^t \left( \frac{2 \eps_t}{5^{t-1}} \right)^{-5}
\le (C/5)^t (50 \eps^2)^{-5} 
= O\big( \eps^{-10} (\log \eps^{-1})^{O(1)} \big)
%^{\log_2 (C/5)} \big)
= O(\eps^{-11}) \, . 
\qedhere
\]
\end{proof}

Combining Theorem~\ref{thm:amplify} with the $\eps$-biased sets constructed in
Section~\ref{sec:homogeneous} we establish a family of
$\eps$-biased set over $G^n$ for smaller $\eps$:

\begin{theorem}
  \label{thm:homogeneous}
  Fix a group $G$. There is an $\eps$-biased set in $G^n$ of size
  $O(n \eps^{-11})$ that can be constructed in time
  polynomial in $n$ and $\eps^{-1}$.
\end{theorem}

\begin{proof}
  Alon et al.~\cite{AlonBNNR:Construction} construct a families of
  explicit codes over finite fields which, in particular, offer 
  $\delta$-biased sets over $\Z_p^n$ of size $O(n)$ for any constant
  $\delta$. As $G$ is fixed, applying Theorem~\ref{thm:constant} to
  these sets over $\Z_{|G|}$ with sufficiently small $\delta \approx 1/\log\log |G|$ yields an $\eps_0$-biased set $S_0$ over $G^n$, where
  $\eps_0$ is a constant close to one (depending on the size of $G$
  and the constant $\delta$). We cannot directly apply
  Theorem~\ref{thm:amplify} to $S_0$, as the bias may exceed
  $1/10$. To bridge this constant gap (from $\eps_0$ to $1/10$),
  we apply the construction of the proof of Theorem~\ref{thm:amplify}
  with a slight adaptation. Selecting a small constant $\alpha$, we may enlarge the expander graph to ensure
  that it has size at least $|S_0| (1/\alpha)$; then the resulting error guarantee on each side of the graph bipartition
  is no more than $\alpha + (1 - \alpha)\eps$ and the product set
  has bias no more than $(\alpha + \eps)^2 + \lambda(\Gamma)$.
  This can be brought as close as desired to $\eps^2$ with
  appropriate selection of the constants $\alpha$ and $\lambda(G)$. As
  $\lambda(G)$ is constant, this transformation likewise increases the
  size of the set by a constant, and this
  method can reduce the error to $1/10$, say, with a constant-factor
  penalty in the size of $S_0$. At this point,
  Theorem~\ref{thm:amplify} applies, and establishes the bound of the theorem.
\end{proof}

\section{Inhomogeneous direct products}
\label{sec:direct}

Groups of the form $G = G_1 \times \cdots \times G_n$ appear to frustrate natural attempts to borrow $\eps$-biased sets directly from abelian groups as we did for $G^n$ in Section~\ref{sec:homogeneous}. In this section, we build an $\eps$-biased set for groups of this form  by iterating a construction that takes $\eps$-biased sets on two groups $G_1$ and $G_2$ and stitches them together, again with an expander graph, to produce an $\eps'$-biased set on $G_1 \times G_2$.  In essence, we again use derandomized squaring, but now for the tensor product of two operators rather than their matrix product.

\begin{construction}\label{cons:direct}
  Let $G_1$ and $G_2$ be two groups; for each $i =1,2$, let $S_i$ be
  an $\eps_i$-biased set on $G_i$. We assume that $|S_1| \leq
  |S_2|$. Let $\Gamma = (U, V; E)$ be a
  bipartite $(|S_2|, d, \lambda)$-expander. Associate elements of $V$
  with elements of $S_2$ and, as in the proof of
  Lemma~\ref{lem:amplify}, associate elements of $S_1$ with $U$ as
  uniformly as possible. As above, we order the elements of $U$ and tile them with copies of $S_1$, leaving a collection of no more than $S_1$ vertices ``uncovered''; these vertices are then assigned to an initial subset of $S_1$ of appropriate size. Define $S_1 \otimes_\Gamma S_2
  \subset G_1 \times G_2$ to be the set of edges of $\Gamma$ (realized
  as group elements according to the association above).
\end{construction}

Recall that an irreducible representation $\rho$ of $G_1 \times G_2$ is a tensor prodoct $\rho_1 \otimes \rho_2$, where each $\rho_i$
is an irrep of $G_i$ and $\rho(g_1, g_2) = \rho_1(g_1) \otimes \rho_2(g_2)$.  If $\rho$ is nontrivial, then one or both of $\rho_1$ and $\rho_2$ is nontrivial, and the bias we achieve on $\rho$ will depend on which of these is the case.

%The bias determined by the construction now depends, in general, on
%the behavior of a representation $\rho$ of $G_1 \times G_2$ on its
%factors. Indeed, the irreducible representations of a direct product
%$G_1 \times G_2$ have the form $\rho_1 \otimes \rho_2$, where each $\rho_i$
%is an irreducible representation of $G_i$ and $[\rho_1\otimes
%\rho_2](g_1, g_2) \eqdef \rho_1(g_1) \otimes \rho_2(g_2)$.

\begin{claim}\label{claim:direct-cons}
  Assuming that $|S_1| \leq |S_2|$, the set $S_1 \otimes_\Gamma S_2$
  of Construction~\ref{cons:direct} has size $d|S_2|$ and bias no more than
  \[
  \max\left(\eps_2, \eps_1 + \frac{|S_1|}{|S_2|}, \lambda + \eps_2\left(\eps_1 + \frac{|S_1|}{|S_2|}\right) \right) \, .
  \]
\end{claim}
\begin{proof}
  The size bound is immediate. As for the bias, let $\rho = \rho_1 \otimes \rho_2$ be nontrivial.  If
  $\rho_1 = \one$, 
  \begin{equation}\label{eq:trivial-good}
  \Bnorm{ \Exp_{s \in S_1 \otimes_\Gamma S_2} (\rho_1 \otimes \rho_2)(s) }
  = \Bnorm{ \Exp_{v \in V} \rho_2(v) }
  \leq \eps_2 \, ,
  \end{equation}
  as $S_2$ is in one-to-one correspondence with $V$. In
  contrast, if $\rho_2 = \one$, the best we can say is that
  \begin{equation}
  \label{eq:trivial-bad}
  \Bnorm{ \Exp_{s \in S_1 \otimes_\Gamma S_2} (\rho_1 \otimes \rho_2)(s) } 
  = \Bnorm{ \Exp_{u \in U} \rho_1(u) } 
  \leq \left( 1 - \frac{|S_1|}{|S_2|} \right) \eps_1 + \frac{|S_1|}{|S_2|}
  \leq \eps_1 + \frac{|S_1|}{|S_2|}
  \end{equation}
  as in the proof of Lemma~\ref{lem:amplify}. When both $\rho_i$ are
  nontrivial, applying Lemma~\ref{lem:expander-operators} to~\eqref{eq:trivial-good} and~\eqref{eq:trivial-bad} implies that
  \begin{equation}\label{eq:nontrivial}
  \Bnorm{ \Exp_{s \in S_1 \otimes_\Gamma S_2} (\rho_1 \otimes \rho_2)(s) } 
  \leq \lambda + \eps_2 \left( \eps_1 + \frac{|S_1|}{|S_2|} \right) \, ,
  \end{equation}
  as desired.
\end{proof}

Finally, we apply Construction~\ref{cons:direct} to groups of the form
$G_1 \times \cdots \times G_n$.

\begin{theorem}\label{thm:direct}
  Let $G = G_1 \times \cdots \times G_n$.  Then, for any
  $\eps$, there is an $\eps$-biased set in $G$ of
  size $\poly(\max_i |G_i|, n, \eps^{-1})$. Furthermore, the set
  can be constructed in time polynomial in its size.
\end{theorem}

\begin{proof}
Given the amplification results of Section~\ref{sec:amplification}, we may focus on constructing sets
  of constant bias. We start by adopting the
  entire group $G_i$ as a $0$-biased set for each $G_i$, and then recursively apply
  Construction~\ref{cons:direct}. This process will only
  involve expander graphs of constant degree, which simplifies the
  task of finding the expander required for
  Construction~\ref{cons:direct}.  In this case, one can construct a
  constant degree expander graph of desired constant spectral gap
  on a set $X$ by covering the vertices of $X$ with a family of
  overlapping expander graphs, uniformizing the degree arbitrarily,
  and forming a small power of the result.  So long as the pairwise
  intersections of the covering expanders are not too small, the
  resulting spectral gap can be controlled uniformly. (This luxury
  was not available to us in the proof of Lemma~\ref{lem:amplify}, since in that setting we required $\lambda$ tending to zero, and insisted on a Ramanujan-like relationship between $\lambda$ and the degree.)
  
The recursive construction proceeds by dividing $G$ into two factors: $A = G_1 \times \cdots \times G_{n'}$ and $B = G_{n'+1} \times \cdots \times G_n$, where $n' = \lceil n/2 \rceil$.  Given small-biased sets $S_A$ and $S_B$, we combine them using Construction~\ref{cons:direct}.  Examining Claim~\ref{claim:direct-cons}, we wish to ensure that $|S_A| / |S_B|$ is a small enough constant. To arrange for this, we assume without loss of generality that $|S_B| \geq |S_A|$ and duplicate $S_B$ five times, resulting in a (multi-)set $S'_B$ such that $|S_A| / |S'_B| \le 1/5$.

Assume that each of the recursively constructed sets $S_A, S_B$ has bias at most $1/4$.  We apply Construction~\ref{cons:direct} to $S_A$ and $S'_B$ with an expander $\Gamma$ of degree $d$ for which $\lambda \leq 1/8$, producing the set $S = S_A \otimes_\Gamma S'_B$.  Ideally, we would like $S$ to also be $1/4$-biased, in which case a set of constant bias and size $\poly(\max_i |G_i|, n)$ would follow by induction.
  
Let $\rho = \rho_A \otimes \rho_B$ be nontrivial, where $\rho_A \in \widehat{A}$ and $\rho_B \in \widehat{B}$.  If $\rho_A = \one$ then, as in~\eqref{eq:trivial-good}, $\bnorm{ \Exp_{s \in S} \rho(s) } \leq 1/4$. Likewise, if both $\rho_A$ and $\rho_B$ are nontrivial, \eqref{eq:nontrivial} gives $\bnorm{ \Exp_{s \in S} \rho(s) } \leq 1/8 + 1/4(1/4 + 1/5) \leq 1/4$.  At first inspection, the case where $\rho_B = \one$ appears problematic, as~\eqref{eq:trivial-bad} only provides the discouraging estimate $\bnorm{ \Exp_{s \in S} \rho(s) } \leq 1/4 + 1/5$.  Thus it seems possible that iterative application of Construction~\ref{cons:direct} could lose control of the error. However, as long as the tiling of $U$, the left side of the expander in Construction~\ref{cons:direct}, is carried out in a way that ensures that the uncovered elements of $U$ are tiled with respect to \emph{previous} stages of the recursive construction, it is easy to check that subsequent recursive appearances of this case can contribute no more than the geometric series $1/5 + (1/5)^2 + \cdots = 1/4$ to the bias. Any following recursive application of the construction in which the representation is nontrivial in both blocks will then drive the error back to $1/4$, as $1/8 + 1/4(1/4 + 1/4) = 1/4$.  (If this case occurs at the last stage of recursion, then $S$ still has bias at most $1/4+1/5 \le 1/2$.)

Recall that for the base case of the induction, we treat each $G_i$ as a $0$-biased set for itself.  Since there are $\log_2 n$ layers of recursion, and each layer multiplies the size of the set by the constant factor $5d$, we end with a $1/2$-biased set $S$ of size at most $(5d)^{\log_2 n} \max_i |G_i| = \poly( \max_i |G_i|, n)$. Finally, applying the amplification of Theorem~\ref{thm:amplify}, after first driving the bias down to $1/10$ as in Theorem~\ref{thm:homogeneous}, completes the proof.
\end{proof}

We note that if the $G_i$ are of polynomial size, then we can use the results of Wigderson and Xiao~\cite{WigdersonX:Derandomizing} to find $\eps$-biased sets of size $O(\log |G_i|)$ in time $\poly(|G_i|)$.  Using these sets in the base case of our recursion then gives a $\eps$-biased set for $G$ of size $\poly(\max_i \log |G_i|, n, \eps^{-1})$.

\section{Normal extensions and smoothly solvable groups}

While applying these techniques to arbitrary groups (even in the case when they have plentiful subgroups) seems difficult, for solvable groups can again use a form of derandomized squaring.  First, recall the derived series: if $G$ is solvable, then setting $G^{(0)}=G$ and taking commutator subgroups $G^{(i+1)} = [G^{(i)},G^{(i)}]$ gives a series of normal subgroups, 
\[
1 = G^{(\ell)} \normal \cdots \normal G^{(1)} \normal G^{(0)} = G \, . 
\]
We say that $\ell$ is the \emph{derived length} of $G$.  
Each factor $G^{(i)}/G^{(i+1)} = A_i$ is abelian, and $G^{(i)}$ is
normal in $G$ for all $i$.  Since $|A_i| \geq 2$, it is obvious that $\ell = O(\log |G|)$.  However, more is true.  The 
\emph{composition series} is a refinement of the derived series where each quotient is a cyclic group of prime order, 
and the length $c$ of this refined series is the \emph{composition length}.  Clearly $c \le \log_2 |G|$.  
Glasby~\cite{glasby} showed that $\ell \le 3 \log_2 c + 9 = O(\log c)$, so $\ell = O(\log \log |G|)$. 

We focus on groups that are \emph{smoothly solvable}~\cite{friedletal}, in the sense that the abelian factors have constant exponent.  (Their definition of smooth solvability allows the factors to be somewhat more general, but we avoid that here for simplicity.)  We then have the following:
\begin{theorem}
\label{thm:solvable}
Let $G$ be a solvable group, 
% with derived length $\ell$, 
and let its abelian factors be of the form $A_i = \Z_{p_i}^t$ (or factors of such groups) where $p_i=O(1)$.  Then $G$ possesses an $\eps$-biased set $S_\eps$ of size 
$(\log |G|)^{1+o(1)} \,\poly(\eps^{-1})$.  
%$\poly(\log |G|, \eps^{-1} )$.
%$\poly((\log |G|)^{\log \log \ell} , \eps^{-1} )$.

\end{theorem}
We deliberately gloss over the issue of explicitness.  However, we claim that if $G$ is polynomially uniform in the sense of~\cite{GenericQFT}, so that we can efficiently express group elements and products as a string of coset representatives in the derived series, then $S_\eps$ can be computed in time polynomial in its size.

\begin{proof}
Solvable groups can be approached via Clifford theory, which controls the structure of representations of a group $G$ when restricted to a normal subgroup.  In fact, we require only a simple fact about this setting.  Namely, if $H \normal G$ and $\rho$ is an irrep of $G$, then either $\Res_H \rho$ contains only copies of the trivial representation so that $\rho(h) = \one_{\rho_d}$ for all $h \in H$, or $\Res_H \rho$ contains \emph{no} copies of the trivial representation.

It is easy to see that the irreps $\rho$ of $G$ for
which $\Res_H \rho$ is trivial are in one-to-one correspondence with
irreps of the group $G/H$, and we will blur this distinction. With this
perspective, it is natural to attempt to assemble an $\eps$-biased set for $G$
from $S_H$, an $\eps_H$-biased set for $H$, and $S_{G/H}$, an
$\eps_{G/H}$-biased set for $G/H$. While $S_H \subset H \subset
G$, there is---in general---no subgroup of $G$ isomorphic to $G/H$, so 
it is not clear how to appropriately embed $S_{G/H}$ into $G$. Happily, we will
see that reasonable bounds can be obtained even with an arbitrary
embedding.  In particular, we treat $S_{G/H}$ as a subset of $G$ by
lifting each element $x \in S_{G/H}$ to an arbitrary element $\hat{x} \in G$ lying in the
$H$-coset associated with $x$.

If $S_H$ and $S_{G/H}$ were the same size, and we could directly
introduce an expander graph $\Gamma$ on $S_H \times S_{G/H}$, then
Lemma~\ref{lem:expander-operators} could still be used to control the
bias of $S = \{ s \hat{t} \mid (s,t) \in \Gamma\}$. Specifically, consider a nontrivial representation $\rho$ of
$G$.  If $\Res_H \rho$ is trivial, then analogous to~\eqref{eq:trivial-good} we have 
$\bnorm{ \Exp_{s \in S} \rho(s) } = \bnorm{ \Exp_{s \in S_{G/H}} \rho(s) } \leq \eps_{G/H}$. 
% (as $\rho(h) = \one$ for any element $h \in H$). 
On the other hand, if $\Res_H
\rho$ restricts to $H$ without any appearances of the trivial
representation, then $\bnorm{ \Exp_{h \in S_H} \rho(h) } \leq \eps_H$. 
In this case, the action of the elements of $S_{G/H}$ on
$\rho$ may be quite pathological, permuting and ``twiddling'' the $H$-irreps 
appearing in $\Res_H \rho$.  
% (especially considering that we permit them to embedded arbitrarily in $G$); 
However, as $\norm{\rho(s)} = 1$ (by unitarity) for all $s \in S_{G/H}$, we can conclude from
Lemma~\ref{lem:expander-operators} that 
$\bnorm{ \Exp_{s \in S} \rho(s) } \leq \lambda(\Gamma) + \eps_H$.

We recursively apply the construction outlined above, accounting for the ``tiling error'' of finding an appropriate expander. Specifically, let us inductively assume we have $\epsilon$-biased sets $S^+$ on $G^{+} = G/G^{(k)}$ and $S^-$ on $G^- = G^{(k)}$ for $k = \lceil \ell/2\rceil$, where $\ell$ is the derived length of $G$.  Selecting an expander graph $\Gamma$ of size at least $\alpha^{-1} \max(|S^-|, |S^+|)$ and
$\lambda(\Gamma) \leq \alpha$, for an $\alpha$ to be determined, we 
tile each side of the graph with elements from $S^-$ and $S^+$,
completing them arbitrarily on the ``uncovered elements.'' 
Since at most a fraction $\alpha$ of the elements on either side are uncovered, the average of a nontrivial representation
over either side of the expander has operator norm no more than $\epsilon + \alpha$.
Lemma~\ref{lem:expander-operators} then implies that the bias of the set $S = \{s \hat{t} \mid (s,t) \in \Gamma\}$ is at most 
$\lambda(\Gamma) + (\epsilon + \alpha) \leq \epsilon + 2\alpha$. 
 If we use the Ramanujan graphs of~\cite{Lubotzky:Ramanujan} described above, we can achieve degree $O(\alpha^{-2})$ and size
$O(\alpha \max(|S^-|, |S^+|))$. Thus, each recursive step of this
process scales the sizes of the sets by a factor $O(\alpha^{-3})$ and
introduces additive error $2\alpha$.  The number of levels of recursion is $\lceil \log_2 \ell \rceil$, 
so if we choose $\alpha < 1/(4 \lceil \log \ell \rceil)$ then the total accumulated error is less than $1/2$.   

Assuming that we have $\alpha$-biased sets for each abelian factor $A_i$ of size no more than $s$, this yields a $1/2$-biased set $S$ for $G$ of size $s \alpha^{-3 \log_2 \ell} = s (\log \ell)^{O(\log \ell)}$.  For constant $p$, there are $\alpha$-biased sets for $\Z_p^n$~\cite{AlonBNNR:Construction} of size $s = O(n / \alpha^3) = (\log |G|)(\log \ell)^{O(1)}$.  Using the fact~\cite{glasby} that $\ell = O(\log \log |G|)$, the total size of $S$ is
\[
(\log |G|) (\log \ell)^{O(\log \ell)}
= (\log |G|) (\log \log \log |G|)^{O(\log \log \log |G|)}
= (\log |G|)^{1+o(1)} \, . 
\]
Finally, we amplify $S$ to an $\eps$-biased set $S_{\eps}$ for whatever $\eps$ we desire with Theorem~\ref{thm:amplify}, introducing a factor $O(\eps^{-11})$.
\end{proof}

%The dependence of $|S|$ on $\log |G|$ in Theorem~\ref{thm:solvable} improves as $\ell$ decreases.  In particular, if $\ell = O(\log \log \ell)$ as for nilpotent groups~\cite[p. 201,216]{Schenkman:Group}, we obtain a constant-bias set of size $s (\log \log |G|)^{\log \log \log |G|} = \poly(\log |G|)$, and thus an $\eps$-biased set $S$ with $|S| = \poly(\log |G|, \eps^{-1})$.

\section*{Acknowledgments}
We thank Amnon Ta-Shma, Emanuele Viola, and Avi Wigderson for helpful discussions.  This work was supported by NSF grant CCF-1117426 and ARO contract W911NF-04-R-0009.  

\bibliography{eps-biased}

\appendix
\section{Quadratic forms associated with expander graphs}
\label{appendix:expanders}

Our goal is to establish the two generalized Rayleigh quotient bounds
described in Lemmas~\ref{lem:sloppy-expander-vectors}
and~\ref{lem:expander-operators}. We begin with the following
preparatory lemma.

\begin{lemma}\label{lem:expander-vectors} Let $G = (U, V; E)$ be a $(n, d, \lambda)$-expander.
  Associate with each vertex $s \in U \cup V$ a vector $\vec{x^s}$ in $\CC^d$ such
  that $\Exp_{u \in U} \vec{x^u} = 0$ and $\Exp_{v \in V} \vec{x^v} = 0$.  Then
  \[
  \left |\Exp_{(u,v) \in E} \langle \vec{x^u}, \vec{x^v}\rangle\right| \leq \lambda \Exp_s \| \vec{x^s}\|^2 \, .
  \]
\end{lemma}

\begin{proof}
  Let $X$ denote the $2n \times d$ matrix whose entries are $X_{sk} =
  x^s_k$.  Then the rows of $X$ are the vectors $\vec{x}$; for an column index
  $k \in \{1, \ldots, d\}$, we let $\vec{y^k} \in \CC^{2n}$ denote the
  vector associated with this column:
  $$
  y^k_v = x^v_k \, .
  $$
  Considering that $\sum_u \vec{x^u} = \sum_v \vec{x^v} = 0$, each
  $\vec{y^k}$ is orthogonal to both $\vec{y^{+}}$ and
  $\vec{y^{-}}$.

  The expectation over a random edge $(u,v)$ of
  $\inner{\vec{x^u}}{\vec{x^v}}$ can be written
  \begin{align*}
    \left| \Exp_{(u,v) \in E} \inner{\vec{x^u}}{\vec{x^v}} \right|
    &= \left| \Exp_{(u,v) \in E} \sum_k X_{uk} X_{vk}\right|
    = \left| \sum_k \Exp_{(u,v) \in E} X_{uk} X_{vk}\right|\\
    &= \left| \sum_k \frac{1}{nd} \sum_{(u,v) \in E} x_k^u x_k^v\right|
    = \left| \frac{1}{n} \sum_k \frac{1}{d} \sum_{(u,v) \in E} y_u^k y_v^k \right|\\
    &= \frac{1}{2n} \left| \sum_k \inner{\vec{y^k}}{A \vec{y^k}} \right|
    \le \frac{1}{2n} \sum_k \left| \inner{\vec{y^k}}{A \vec{y^k}}\right|\\
    &\le \frac{\lambda}{2n} \sum_k \norm{\vec{y^k}}^2 =
    \frac{\lambda}{2n} \sum_s \norm{\vec{x^s}}^2
    = \lambda \Exp_s \| \vec{x^s}\|^2\, . \qedhere
  \end{align*}
\end{proof}

\begin{lemma}  \label{lem:sloppy-expander-vectors}
  Let $G = (U, V; E)$ be a $(n, d, \lambda)$-expander.  Associate
  with each vertex $s \in U \cup V$ a vector $\vec{x^s}$ in $\CC^d$
  such that $\|\Exp_{u \in U} \vec{x^u}\| =
  \eps_U$ and $\| \Exp_{v \in V} \vec{x^v}\| = \eps_V$.  Then
  \[
  \left |\Exp_{(u,v) \in E} \langle \vec{x^u}, \vec{x^v}\rangle\right|
  \leq \lambda \left(\Exp_s \| \vec{x^s}\|^2 - \frac{ \eps_U^2}{2}-\frac{\eps_V^2}{2}\right) + \eps_U \eps_V \, .
  \]
\end{lemma}

\begin{proof}[Proof of Lemma~\ref{lem:sloppy-expander-vectors}]
Let $\vec{x}^U = \Exp_{u \in U} \vec{x^u}$ and $\vec{x}^V = \Exp_{v
  \in V} \vec{x^v}$. We have
\begin{align*}
  \left |\Exp_{(u,v) \in E} \langle \vec{x^u}, \vec{x^v}\rangle\right|
  = \left |\Exp_{(u,v) \in E} \langle (\vec{x^u} - \vec{x^U}) +
    \vec{x^U}, (\vec{x^v} - \vec{x^V}) + \vec{x^V}\rangle\right|
\end{align*}
which we may further expand into
\begin{equation}
\label{eq:exp}
\left |\Exp_{(u,v) \in E} \langle (\vec{x^u} - \vec{x^U}),
    (\vec{x^v} - \vec{x^V})\rangle
+ \Exp_{(u,v) \in E} \langle \vec{x^U},
    (\vec{x^v} - \vec{x^V})\rangle
+ \Exp_{(u,v) \in E} \langle (\vec{x^u}-\vec{x^U}),
   \vec{x^V}\rangle + \Exp_{(u,v) \in E} \langle \vec{x^U},
    \vec{x^V}\rangle\right| \, .
\end{equation}
As $G$ is regular, the vertices of a uniformly random edge $(u,v)$
are individually uniform on $U$ and $V$, from which it follows that
the two middle terms of~\eqref{eq:exp} are both zero. Hence we
conclude that
\[
\left|\Exp_{(u,v) \in E} \langle \vec{x^u}, \vec{x^v}\rangle\right|
\leq \left|\Exp_{(u,v) \in E} \langle (\vec{x^u} - \vec{x^U}),
    (\vec{x^v} - \vec{x^V})\rangle\right| + 
\bigl|\langle \vec{x^U},
    \vec{x^V}\rangle\bigr| \, .
\]

Applying Lemma~\ref{lem:expander-vectors}
to the the vectors
\[
\vec{x^u} -\vec{x^U} \quad \text{and}\quad
\vec{x^v} -\vec{x^V} \, ,
\] 
we conclude that
\[
\left|\Exp_{(u,v) \in E} \langle (\vec{x^u} - \vec{x^U}),
    (\vec{x^v} - \vec{x^V})\rangle\right| \leq \frac{\lambda}{2n}\left(\sum_u\norm{\vec{x^u}-\vec{x^U}}^2+
    \sum_v\norm{\vec{x^v}-\vec{x^V}}^2\right) \, .
\]    
    
    The summation $\sum_u\norm{\vec{x^u}-\vec{x^U}}^2$ can be calculated as follows.
\begin{align*}
  \sum_u\norm{\vec{x^u}-\vec{x^U}}^2 
  &= \sum_u\inner{\vec{x^u}-\vec{x^U}}{\vec{x^u}-\vec{x^U}}\\
  &= \sum_u(\inner{\vec{x^u}}{\vec{x^u}}-\inner{\vec{x^u}}{\vec{x^U}}-\inner{\vec{x^U}}{\vec{x^u}}+\inner  {\vec{x^U}}{\vec{x^U}})\\
  &= \sum_u\inner{\vec{x^u}}{\vec{x^u}}-n\inner{\vec{x^U}}{\vec{x^U}}-n\inner{\vec{x^U}}{\vec{x^U}}+n\inner{\vec{x^U}}{\vec{x^U}}\\
  &= \sum_u\norm{\vec{x^u}}^2-n\norm{\vec{x^U}}^2\\
  &= \sum_u\norm{\vec{x^u}}^2-n\eps_U^2 \, . 
\end{align*}

Therefore, 
\begin{align*}
\left|\Exp_{(u,v) \in E} \langle (\vec{x^u} - \vec{x^U}),
    (\vec{x^v} - \vec{x^V})\rangle\right| &\leq
  \frac{\lambda}{2n}\left(\sum_u\norm{\vec{x^u}}^2-n\eps_U^2+\sum_v\norm{\vec{x^v}}^2-n\eps_V^2\right)\\
  &\leq \lambda \left(\Exp_s \| \vec{x^s}\|^2 - \frac{ \eps_U^2}{2}-\frac{\eps_V^2}{2}\right) \, .
\end{align*}
By Cauchy-Schwarz, we have  $|\langle \vec{x^U},
    \vec{x^V}\rangle| \leq \eps_U \eps_V$. 
In total, then,
\[
\left|\Exp_{(u,v) \in E} \langle \vec{x^u}, \vec{x^v}\rangle\right|
\leq \lambda \left(\Exp_s \| \vec{x^s}\|^2 - \frac{ \eps_U^2}{2}-\frac{\eps_V^2}{2}\right) + \eps_U\eps_V \, ,
\]
as desired.
\end{proof}

\begin{lemma-restatement}[Restatement of Lemma~\ref{lem:expander-operators}] Let $G = (U \cup V, E)$ be a $(n, d, \lambda)$-expander.
  Associate with each vertex $s \in U \cup V$ a linear operator $X_s$
  on the vector space $\CC^d$ such
  that $\| X_s\| \leq 1$, $\|\Exp_{u \in U} X_u\| = \eps_U$, and $\|\Exp_{v \in V} X_v\| = \eps_V$.  Then
  \[
  \left\| \Exp_{(u,v) \in E} X_u X_v \right\| \leq \lambda  + (1
  - \lambda)\eps_U\eps_V \, .
  \]
\end{lemma-restatement}

\begin{proof}[Proof of Lemma~\ref{lem:expander-operators}]
Let $X$ denote the linear operator $\Exp_{(u,v) \in E} X_u X_v$. Writing
\[
\| X \| = \max_{\substack{\|\vec{x}\|=1,\\ \|\vec{y}\|=1}} | \langle \vec{x},
X\vec{y}\rangle| \, ,
\]
we observe that
\[
\langle \vec{x},
X\vec{y}\rangle = \left\langle \vec{x}, \Exp_{(u,v) \in E} X_u X_v
\vec{y}\right\rangle = \Exp_{(u,v) \in E} \langle X_u^\dagger \vec{x}, X_v
\vec{y}\rangle \, .
\]
Considering the bounds on $\Exp_u X_u$ and $\Exp_v X_v$, it follows
that $\|\Exp_u X_u^\dagger \vec{x}\| \leq \eps_U$ and $\|\Exp_v
X_v \vec{y}\| \leq \eps_V$; applying Lemma~\ref{lem:sloppy-expander-vectors} with the vector family
$\vec{x^u} = X_u^\dagger \vec{x}$ and $\vec{x^v} = X_v \vec{y}$ we
conclude that
\begin{align*}
|\langle \vec{x},
X\vec{y}\rangle|  &\leq \max_{\substack{\delta_U \leq \eps_U\\
    \delta_V \leq \eps_V}} \lambda \left(\Exp_s \| \vec{x^s}\|^2 - \frac{ \delta_U^2}{2}-\frac{\delta_V^2}{2}\right) +
\delta_U\delta_V\\
&\leq \max_{\substack{\delta_U \leq \eps_U\\
    \delta_V \leq \eps_V}} \lambda \left(1 - \delta_U\delta_V\right) +
\delta_U\delta_V \leq \lambda +
(1 - \lambda)\eps_U\eps_V
\end{align*}
as $\delta_U^2 + \delta_V^2 \geq 2 \delta_U \delta_V$.
\end{proof}

\section{A tail bound for products of operator-valued random variables}
\label{app:tail}

Our goal is to establish the following tail bound (a restatement and
expansion of
Theorem~\ref{thm:operator-avg}).

\begin{theorem-restatement}[Restatement of Theorem~\ref{thm:operator-avg}]
  Let $\positive(H)$ denote the cone of positive operators on the
  Hilbert space $H$ and let $P_1,\ldots, P_k$ be independent random variables taking values
  in $\positive(H)$ for which
  \[
  \norm{P_i} \leq 1 \, ,\quad \text{and}\quad \bnorm{ \Exp[P_i] } \leq 1 - \delta \, .
  \]
  Then for any $\Delta \geq 0$, 
  \[
  \Pr\left[ \bnorm{P_k \cdots P_1} 
  \geq \sqrt{\dim H} \exp\left(-\frac{k\delta}{2} + \Delta\right)\right] 
  \leq \dim H \cdot \exp\left(-\frac{\Delta^2}{2k\ln 2} \right) \, .
  \]
  In particular, choosing $\Delta = k\delta/3$, we conclude that
  \[
  \Pr\left[ \bnorm{ P_k \cdots P_1 } \geq
    \sqrt{\dim H} \exp\left(-\frac{k\delta}{6}\right)\right] \leq
  \dim H \cdot \exp\left(-\frac{k\delta^2}{18\ln 2} \right) \leq
  \dim H \cdot \exp\left(-\frac{k\delta^2}{13} \right) \, . \qedhere
  \]
\end{theorem-restatement}

Recall Azuma's inequality for supermartingales:
\begin{theorem}[Azuma's inequality] Let $X_0, \ldots, X_T$ be a family
  of real-valued random variables for which $|X_i - X_{i-1}| \leq \alpha_i$ and
  $\Exp[X_i \mid X_1, \ldots, X_{i-1}] \leq X_{i-1}$. Then
\[
\Pr[X_T - X_0 \geq \lambda] \leq \exp\left(-\frac{\lambda^2}{2\sum_i
    \alpha_i}\right) \, .
\]
\end{theorem}

\begin{corollary}\label{cor:azuma}Let $X_0, \ldots, X_T$ be a family
  of real-valued random variables for which $ X_{i-1}  - \alpha_i \leq
  X_i \leq X_{i-1}$ and
  $\Exp[X_i \mid X_1, \ldots, X_{i-1}] \leq X_{i-1} - \eps_i$ for
  some $\eps_i \leq \alpha_i$. Then
\[
\Pr[X_T - X_0 \geq -\sum_i \eps_i  + \lambda] \leq
\exp\left(-\frac{\lambda^2}{2\sum_i \alpha_i}\right) \, .
\]
\end{corollary}

\begin{proof}
  Apply Azuma's inequality to the random variables $\tilde{X_t} = X_t
  + \sum_{i}^t \eps_i$.
\end{proof}

\begin{proof}[Proof of Theorem~\ref{thm:operator-avg}]
  We begin by considering the behavior of the operator $P_k\cdots P_1$
  on a particular vector $\vec{v}$. To complete the proof we will
  select an orthonormal basis $\mathcal{B}$ of $H$.  The operator norm is bounded above by the 
  Frobenius norm, 
  \begin{equation}
    \label{eq:op-inequality}
  \norm{P_k \cdots P_1} 
  \leq \norm{P_k \cdots P_1}_F
  = \sqrt{\sum_{\vec{b} \in \mathcal{B}} \| P_k \ldots P_1\vec{b} \|^2} 
  \le \sqrt{\dim H} \cdot \max_{\vec{b} \in \mathcal{B}} \| P_k \ldots P_1\vec{b}\| \, .
  \end{equation}
  Now fix a unit-length vector $\vec{v} \in H$ and consider the random variables
  \[
  \vec{v_0} = \vec{v}, \quad \vec{v_1} 
  = P_1\vec{v}, \quad \vec{v_2} 
  = P_2 P_1 \vec{v}, 
  \]
  and
  \[
  \ell_i = \begin{cases} {\|\vec{v_i}\|}/{\|\vec{v_{i-1}}\|} & \text{if $\vec{v_{i-1}} \neq 0$,}\\
    0 & \text{otherwise,}
  \end{cases}
  \]
  Our goal is to establish strong tail bounds on the random variable $\| \vec{v_k}\| = \ell_k \ell_{k-1} \ldots \ell_1$.
  Recalling that $\norm{\Exp[P_i]} \leq 1-\eps$ and that the $P_i$ are independent we have 
  \begin{equation}\label{eq:martingale}
  \Exp\left[ \ell_i \;\middle|\; P_1, \ldots,
    P_{i-1}\right] \leq 1 - \eps \, , 
  \end{equation}
  and we proceed to apply a martingale tail bound.

  It will be more convenient to work with log-bounded random
  variables, so we define $m_i = \max(\ell_i, 1/2)$ and observe that
 $\norm{\vec{v_k}} 
 %= \ell_k \ell_{k-1} \cdots \ell_i \| \vec{v_0} \| 
 \leq m_k m_{k-1} \ldots m_1$
 and 
  $\ln \| \vec{v_k}\| \leq \sum_i \ln m_i$.  
  Considering that $\max(x,1/2) \leq (1 + x)/2$ for $x \in [0,1]$ we conclude 
  from equation~\eqref{eq:martingale} above that
  $\Exp\left[ m_i \;\middle|\; P_1, \ldots,
    P_{i-1}\right] \leq 1 - \eps/2$.
  Since $1/2 \leq m_i \leq 1$ and $\ln m \le m-1$ for $m > 0$, 
%  , considering that
%  \[
%  \ln(1 - x) = - x - \frac{x^2}{2} - \frac{x^3}{3} - \cdots \leq -x
%  \]
%  for $0 \leq x < 1$, 
 we have
 \begin{equation}\label{eq:m-martingale}
    \Exp\left[ \ln m_i \;\middle|\; P_1, \ldots,
       P_{i-1}\right] \leq - {\eps}/{2} \, .
  \end{equation}

  Applying Azuma's inequality (specifically, Corollary~\ref{cor:azuma} above) to the random variables 
  $M_t = \sum_{i=1}^t \ln m_i$, we conclude that
  \[
  \Pr\left[ M_k \geq -\frac{k\eps}{2} + \Delta\right] =
  \Pr\left[ \sum_i \ln m_i \geq -\frac{k\eps}{2} +
    \Delta\right] \leq \exp\left(-\frac{\Delta^2}{2k\ln 2} \right)
  \]
  and hence
  \[
  \Pr\left[ \|P_k\cdots P_1 \vec{v}\| \geq \exp\left(-\frac{k\eps}{2} + \Delta\right)\right] \leq \exp\left(-\frac{\Delta^2}{2k\ln 2} \right) \, .
  \]
  Applying the above inequality to an orthonormal basis $\vec{b_1}, \ldots,
  \vec{b_n}$ of $H$, we find that
  \[
  \Pr\left[ \exists i: \bigl\|P_k \cdots P_1 \vec{b_i}\bigr\|_2 \geq
    \exp\left(-\frac{k\eps}{2} + \Delta\right)\right] \leq
  \dim H \cdot \exp\left(-\frac{\Delta^2}{2k\ln 2} \right)
  \]
  by the union bound.  Applying~\eqref{eq:op-inequality} then gives 
%  As $\| A\| \leq \|A\|_F \leq
%  \sqrt{\sum_{\vec{b} \in \mathcal{B}} \|A\vec{b}\|^2} \leq 
%  \sqrt{\dim H} \cdot \max_{\vec{b} \in \mathcal{B}} \| A\vec{b}\|$, we have
  \[
  \Pr\left[ \bigl\|P_k \cdots P_1\bigr\| \geq \sqrt{\dim H}\exp\left(-\frac{k\eps}{2} + \Delta\right)\right] \leq
  \dim H \cdot \exp\left(-\frac{\Delta^2}{2k\ln 2} \right) \, .\qedhere
  \]
 \end{proof}

\end{document}